\documentclass[letterpaper, 10 pt, conference]{ieeeconf}  

\IEEEoverridecommandlockouts                              
\usepackage{amsmath}
\usepackage{amsfonts}
\usepackage{mathrsfs}
\usepackage{algorithmic}
\usepackage{algorithm}
\usepackage{array}
\usepackage{comment}
\usepackage{bm}
\usepackage{graphicx}
\usepackage{booktabs}
\usepackage{hyperref}
\usepackage{color}

\newtheorem{theorem}{Theorem}[section]

\newtheorem{lemma}{Lemma}[section]

\newtheorem{remark}{Remark}[section]

\title{\LARGE \bf
Maximum Correntropy Ensemble Kalman Filter
}

\author{Yangtianze Tao, Jiayi Kang and Stephen Shing-Toung Yau
\thanks{This work is supported by National Natural Science Foundation of China (NSFC) grant (11961141005) and Tsinghua University Education Foundation fund (042202008). \emph{Corresponding author: Stephen Shing-Toung Yau.}}
\thanks{Yangtianze Tao and Jiayi Kang are with the Department of Mathematical Sciences, Tsinghua University, Beijing 100084, China. 
{Email: \tt\small tytz19@mails.tsinghua.edu.cn, \tt\small kangjy19@mails.tsinghua.edu.cn}}
\thanks{Stephen Shing-Toung Yau is with the Department of Mathematical Sciences, Tsinghua University, Beijing 100084, China, and with the Yanqi Lake Beijing Institute of Mathematical Sciences and Applications, Huairou district, Beijing 101400, China.
{Email: \tt\small yau@uic.edu}}
}

\begin{document}

\maketitle
\thispagestyle{empty}
\pagestyle{empty}

\begin{abstract}
In this article, a robust ensemble Kalman filter (EnKF) called MC-EnKF is proposed for nonlinear state-space model to deal with filtering problems with non-Gaussian observation noises. Our MC-EnKF is derived based on maximum correntropy criterion (MCC) with some technical approximations. Moreover, we propose an effective adaptive strategy for kernel bandwidth selection.
Besides, the relations between the common EnKF and MC-EnKF are given, i.e., MC-EnKF will converge to the common EnKF when the kernel bandwidth tends to infinity. This justification provides a complementary understanding of the kernel bandwidth selection for MC-EnKF. In experiments, non-Gaussian observation noises significantly reduce the performance of the common EnKF for both linear and nonlinear systems, whereas our proposed MC-EnKF with a suitable kernel bandwidth maintains its good performance at only a marginal increase in computing cost, demonstrating its robustness and efficiency to non-Gaussian observation noises.
\end{abstract}

\section{Introduction}
Filtering \cite{sarkka2013bayesian} is the fundamental issue for the state estimation of the state-space model \cite{hamilton1994state}. It has a large number of applications in many fields, such as robot vision \cite{chen2011kalman} and data assimilation \cite{law2015data}. For the linear state-space model with Gaussian noise, the optimal filtering solution can be computed analytically, which is the well-known Kalman filter (KF) \cite{kalman_first}. Since then, a large number of filtering algorithms for nonlinear state-space model with Gaussian noise have been proposed. The representatives among them are extended Kalman filter (EKF) \cite{hoshiya1984structural}, unscented Kalman filter (UKF) \cite{julier1997new}, cubature Kalman filter (CKF) \cite{arasaratnam2009cubature} and ensemble Kalman filter (EnKF) \cite{evensen2003ensemble}.

In the case of Gaussian noise, the KF and its expansions work well. However, the noise frequently does not fit the Gaussian distribution in real-world application circumstances. For instance, in many practical settings of target tracking \cite{lu2022measurement} and power systems \cite{wang2019adaptive}, impulsive interferences and observation outliers are frequent.
These disturbances are often modelled by some heavy-tailed impulsive noises (such as
some mixed-Gaussian distributions). The main
reason for this problem is that the KF is based on the well-known
minimum mean square error (MMSE) criterion \cite{jazwinski2007stochastic}, which is sensitive to large outliers and deteriorates the KF's robustness in non-Gaussian noise \cite{wu2015kernel}. Numerous studies have attempted to address the filtering problems with non-Gaussian noise, such as filters based on information theoretic quantities \cite{principe2010information}, the Huber-based KFs \cite{stojanovic2016robust} and robust Student’s t filter \cite{huang2017novel}. Besides these filters, a local similarity measure called correntropy \cite{liu2007correntropy} has been used to develop new robust filters \cite{chen2017maximum, wang2017maximum, wang2019iterated, tao2023outlier}. 
These filters are derived by maximum correntropy criterion (MCC) and called MCC-based filters, which can achieve great performance in the presence of heavy-tailed non-Gaussian noises since correntropy is insensitive to large outliers.

In this article, we formulate the estimation problem under MCC framework. Under this framework, the goal of this article is to develop a robust EnKF called the maximum correntropy EnKF (MC-EnKF) based on the MCC, to address nonlinear filtering problems when observations contain large outliers. With the help of suitable approximation techniques, we derive the recursion of ensembles for MC-EnKF based on MCC cost function. It is necessary to note that MC-EnKF employs the empirical prediction covariance calculated by ensembles. This procedure is consistent with the common EnKF, which allows us to prove that MC-EnKF will
converge to the common EnKF when the kernel bandwidth
tends to infinity. This justification provides a complementary
knowledge of understanding the kernel bandwidth selection
for MC-EnKF. Although MC-EnKF aims to handle non-Gaussian noise cases, it also has the flexibility to handle
Gaussian noise cases with a large enough kernel bandwidth
since EnKF performs well in this scenario. The contributions of this article are listed as follows.
\begin{itemize}
    \item We introduce the idea of the cost function to the field of EnKF, which inspires us to derive a robust EnKF called MC-EnKF based on the MCC cost function. Moreover, we propose an effective adaptive strategy for kernel bandwidth selection. Besides, we provide a complementary understanding of the kernel bandwidth selection theoretically, i.e., MC-EnKF will converge to the common EnKF when the kernel bandwidth is large enough. This justification gives the flexibility of MC-EnKF to handle cases involving both Gaussian and non-Gaussian noise.
    
    \item Our theoretical result is supported by numerical experiments on several filtering benchmarks, i.e., our proposed MC-EnKF will perform like the common EnKF when the kernel bandwidth is large enough. Furthermore, experiments demonstrate that our proposed adaptive strategy for kernel bandwidth selection is effective. And with the appropriate kernel bandwidth, MC-EnKF can outperform EnKF with only a minor increase in computing cost when the underlying observation system is hampered by some heavy-tailed non-Gaussian noises.
\end{itemize}

The remainder of this article is organized as follows. In section \ref{sec:background}, we present the concept of MCC, and briefly review the nonlinear filtering problems and EnKF algorithm. In section \ref{sec:method}, we derive our MC-EnKF algorithm and propose an effective adaptive strategy for kernel bandwidth selection with necessary discussions. The experiments and discussions are presented in Section \ref{sec:experiments}. The conclusions are drawn in Section \ref{sec:conclusion}.

\subsection{Notations}
Throughout the article, we use boldface lower-case letters
for vectors and boldface upper-case letters for matrices. The transpose and mathematical expectation are denoted by $\{\cdot\}^{\top}$ and $\mathbb{E}\left[\cdot\right]$, respectively. The Gaussian distribution
with mean $\mu$ and covariance $\bm{\Sigma}$ is denoted by $\mathcal{N}(\mu, \bm{\Sigma})$. $\mathbb{R}^{n}$ and $\mathbb{R}^{m\times n}$ denote, respectively, the $n$-dimensional
Euclidean space and the set of all $n\times m$ real matrices. Especially, $\mathbf{I}_{n}$ denotes the n-dimensional identity matrix. weighted $l^{2}$ norm $\|\mathbf{x}\|_{\mathbf{A}}^{2} := \mathbf{x}^{\top}\mathbf{A}^{-1}\mathbf{x}$ with $\mathbf{x}\in\mathbb{R}^{n}$ and non-singular matrix $\mathbf{A}\in\mathbb{R}^{n\times n}$.
   
\section{Background}\label{sec:background}
In this section, we shall first introduce the concept of correntropy \cite{liu2007correntropy}, which is a local similarity measure between two random vectors. Then we shall formulate the estimation problems based on MCC and MMSE estimate \cite{jazwinski2007stochastic}. At last, we shall briefly review the framework of nonlinear filtering problems and introduce the EnKF algorithm.

\subsection{Maximum Correntropy Criterion and Minimum Mear Square Error}

\subsubsection{Correntropy}
For given two random vectors $\mathbf{x}$ and $\mathbf{y}$, the correntropy between $\mathbf{x}$ and $\mathbf{y}$ is denoted by $\mathcal{V}(\mathbf{x}, \mathbf{y})$, which is defined as follows:
\begin{equation}
\mathcal{V}(\mathbf{x}, \mathbf{y})=\mathbb{E}
\left[\mathcal{G}_\sigma(\|\mathbf{x} - \mathbf{y}\|_{\mathbf{A}_{1}})\right],
\end{equation}
where $\mathbf{A}_{1}$ is a given non-singular matrix with suitable dimension, $\mathcal{G}_\sigma$ is the Gaussian Kernel given by
$\mathcal{G}_\sigma(e)=\exp \left(-\frac{e^2}{2 \sigma^2}\right)$, and $\sigma>0$ stands for the kernel bandwidth. Here we shall recall the important property for $\mathcal{V}(\mathbf{x}, \mathbf{y})$, more details can be found in \cite{liu2007correntropy}. By taking the Taylor series expansion of the Gaussian Kernel, we have
\begin{equation}
\mathcal{V}(\mathbf{x}, \mathbf{y})=\sum_{n=0}^{\infty} \frac{(-1)^n}{2^n \sigma^{2 n} n !} \mathbb{E}\left[\|\mathbf{x}-\mathbf{y}\|_{\mathbf{A}_{1}}^{2 n}\right].
\end{equation}
The correntropy can be seen as the weighted sum of all even order moments of the error vectors $\mathbf{x}-\mathbf{y}$. The parameter to weight the second order and higher order moments appears to be the kernel bandwidth. The second order moment will predominate in the correntropy when the kernel bandwidth is particularly large compared to the error vectors range.

\subsubsection{Estimation Problems Formulation}
Correntropy can be used as the optimality criterion for estimation problems. Suppose our goal is to learn a parameter $\theta$ for a given estimator $\mathbf{x}(\theta)$, and let $\mathbf{y}$ denote the desired output. Then the MCC-based estimation problem can be formulated as solving the following optimzation problem:
\begin{equation}\label{eqn:mcc-cost}
    \hat{\bm{\theta}}_{1} = \arg\max_{\bm{\theta}\in\Omega_{1}} \ \mathcal{V}(\mathbf{x}(\theta), \mathbf{y}),
\end{equation}
where $\hat{\bm{\theta}}_{1}$ denotes its optimal solution and $\Omega_{1}$ denotes its feasible set for parameters. It is necessary to compare MCC with the conventional MMSE criterion. Let $\Omega_{2}$ denotes its feasible set for parameters with $\hat{\bm{\theta}}_{2}$ as optimal solution. Then for the given non-singular matrix $\mathbf{A}_{2}$ with suitable dimension, the MMSE-based estimation problem is formulated as follows:
\begin{equation}\label{eqn:mmse-cost}
\begin{aligned}
     \hat{\bm{\theta}}_{2} & = \arg\max_{\bm{\theta}\in\Omega_{2}} \ -\|\mathbf{x}(\theta) - \mathbf{y}\|_{\mathbf{A}_{2}} \\
     & = \arg\min_{\bm{\theta}\in\Omega_{2}} \ \|\mathbf{x}(\theta) - \mathbf{y}\|_{\mathbf{A}_{2}}.
\end{aligned}
\end{equation}
\subsection{Nonlinear Filtering Problems}
In this article, we consider the nonlinear state-space model given by the following state and observation equations:
\begin{equation}\label{eqn:state-space-model}
    \begin{aligned}
        \mathbf{x}_{k} & = \mathbf{f}_{k}(\mathbf{x}_{k-1}) + \mathbf{w}_{k} \quad  \\
        \mathbf{y}_{k} & = \mathbf{h}_{k}(\mathbf{x}_{k}) + \mathbf{v}_{k},
    \end{aligned}
\end{equation}
where state $\mathbf{x}_{k}\in\mathbb{R}^{n}$, and observation $\mathbf{y}_{k}\in \mathbb{R}^{m}$. $\mathbf{f}_{k}$ and $\mathbf{h}_{k}$ are nonlinear functions called state equation and observation equation, respectively. State noise $\mathbf{w}_{k}$ and observation noise $\mathbf{v}_{k}$ are zero means with nominal covariance $\mathbf{Q}_{k}\in\mathbb{R}^{n\times n}$ and $\mathbf{R}_{k}\in\mathbb{R}^{m\times m}$, respectively. Let $\mathbf{y}_{1:k}$ denote the $\sigma$-algebra generated by noisy observations $\{\mathbf{y}_{1}, \dots, \mathbf{y}_{k}\}$. The filtering problem refers to estimating the posterior distribution $p(\mathbf{x}_{k}\mid\mathbf{y}_{1:k})$, which is called filtering distribution. 

 In the follow-up, we focus on the scenario in which the observation noises are not Gaussian, i.e., they are induced by unknown large outliers. Hence, the real distribution of observation noise is unknown to us in fact. For example, we consider $(1-\epsilon)\ \mathcal{N}(0, \mathbf{R}_{k}) + \epsilon \ \mathcal{S}(0, \mathbf{S}_{k})$ , where $0<\epsilon\ll 1$ is the unknown probability, and $\mathcal{S}(0, \mathbf{S}_{k})$ is an arbitrary unknown
distribution with large covariance $\mathbf{S}_{k}$. Here $\mathbf{R}_{k}$ is known to us, so it is called the nominal covariance matrix. Additionally, in what follows, we shall develop our novel filter to deal with such noises above by utilizing the MCC as the cost function involving the nominal observation covariance $\mathbf{R}_{k}$. 

\subsection{Ensemble Kalman Filter}

The EnKF sequentially approximates the filtering distributions $p(\mathbf{x}_{k}\mid \mathbf{y}_{1:k})$ using $N$ equally weighted ensembles $\{\mathbf{x}_{k}^{(1)}, \dots, \mathbf{x}_{k}^{(N)}\}$. At prediction steps, each ensemble $\mathbf{x}_{k}^{(i)}$
is propagated using the state equation, while at update steps a Kalman-type update is performed for each ensemble:

\begin{equation}\label{eqn:enkf-pred}
    \mathbf{x}_{k\mid k-1}^{(i)}  = \mathbf{f}_{k}(\mathbf{x}_{k-1\mid k-1}^{(i)}) + \mathbf{w}_{k}^{(i)},
\end{equation}
and 
\begin{equation}\label{eqn:enkf-update}
     \mathbf{x}_{k\mid k}^{(i)} = \mathbf{x}_{k\mid k-1}^{(i)} + \hat{\mathbf{K}}_{k} \left(\mathbf{y}_{k} + \mathbf{v}_{k}^{(i)} - \mathbf{h}_{k}(\mathbf{x}_{k\mid k-1}^{(i)}) \right),
\end{equation}
where $\mathbf{w}_{k}^{(i)}\sim\mathcal{N}(0, \mathbf{Q}_{k})$, $\mathbf{v}_{k}^{(i)}\sim\mathcal{N}(0, \mathbf{R}_{k})$ and the Kalman gain 
\begin{equation}\label{eqn:enkf-gain}
\hat{\mathbf{K}}_{k}=\hat{\mathbf{C}}_{k}\mathbf{H}_{k}^{\top}(\mathbf{H}_{k}\hat{\mathbf{C}}_{k}\mathbf{H}_{k}^{\top}+\mathbf{R}_{k})^{-1},
\end{equation}
is defined using the empirical prediction covariance $\hat{\mathbf{C}}_{k}$
of prediction ensembles $\{\mathbf{x}_{k\mid k-1}^{(1)}, \dots, \mathbf{x}_{k\mid k-1}^{(N)}\}$, namely
\begin{equation}\label{eqn:enkf-cov}
    \hat{\mathbf{C}}_{k} = \frac{1}{N-1}\sum_{i=1}^{N}(\mathbf{x}_{k\mid k-1}^{(i)} - \hat{\mathbf{m}}_{k})(\mathbf{x}_{k\mid k-1}^{(i)} - \hat{\mathbf{m}}_{k})^{\top}, 
\end{equation}
with 
\begin{equation}\label{eqn:enkf-mean}
    \hat{\mathbf{m}}_{k} = \frac{1}{N}\sum_{i=1}^{N}\mathbf{x}_{k\mid k-1}^{(i)},
\end{equation}
and 
\begin{equation}\label{eqn:approx-H}
\mathbf{H}_{k}=\frac{\partial\mathbf{h}_{k}}{\partial\mathbf{x}_{k}}\mid_{\mathbf{x}_{k}=\hat{\mathbf{m}}_{k}}.    
\end{equation}
\begin{remark}\label{remark:enkf-prediction-approx}
Empirical prediction mean in \eqref{eqn:enkf-mean} and empirical prediction covariance in \eqref{eqn:enkf-cov} provide a Gaussian approximation to the prediction distribution distribution, i.e.,
\begin{equation}
    p(\mathbf{x}_{k}\mid \mathbf{y}_{1:k-1})\approx\mathcal{N}(\hat{\mathbf{m}}_{k}, \hat{\mathbf{C}}_{k}).
\end{equation}
\end{remark}

\section{Proposed Algorithm}\label{sec:method}

In this section, we shall present the derivation of MC-EnKF. Then we shall give the discussions on the adaptive strategy for kernel bandwidth selection and the convergence of MC-EnKF with respect to kernel bandwidth.

\subsection{Derivation of the Algorithm}
Here we present how to derive MC-EnKF based on the MCC estimate \eqref{eqn:mcc-cost}. The prediction step of MC-EnKF is the same as those of EnKF, i.e., \eqref{eqn:enkf-pred}. Therefore, we shall focus on deriving its update step. Let $\hat{\mathbf{x}}_{k\mid k-1}$ and $\mathbf{P}_{k\mid k-1}$ denote the prediction mean $\mathbb{E}\left[\mathbf{x}_{k}\mid \mathbf{y}_{1:k-1}\right]$ and prediction covariance $\mathbb{E}\left[\left(\mathbf{x}_{k} - \hat{\mathbf{x}}_{k\mid k-1}\right)\left(\mathbf{x}_{k} - \hat{\mathbf{x}}_{k\mid k-1}\right)^{\top}\mid \mathbf{y}_{1:k-1}\right]$, respectively.
As shown in \cite{rauch1965maximum}, for the underlying linear system of \eqref{eqn:state-space-model}, i.e., $\mathbf{f}_{k}(\mathbf{x}_{k})=\mathbf{F}_{k}\mathbf{x}_{k}$ and $\mathbf{h}_{k}(\mathbf{x}_{k})=\mathbf{H}_{k}\mathbf{x}_{k}$, the update step of KF can be derived by using least square cost function, 
\begin{equation}\label{eqn:kf-cost}
    \mathcal{L}_{1}(\mathbf{x}_{k}) = \|\mathbf{y}_{k}-\mathbf{H}_{k}\mathbf{x}_{k}\|_{\mathbf{R}_{k}}+\|\mathbf{x}_{k}-\hat{\mathbf{x}}_{k\mid k-1}\|_{\mathbf{P}_{k\mid k-1}}.
\end{equation}
Motivated by \eqref{eqn:kf-cost}, we shall consider a new cost function based on MCC. Besides, in view of Remark \ref{remark:enkf-prediction-approx}, we shall replace $\hat{\mathbf{x}}_{k\mid k-1}$ and $\mathbf{P}_{k\mid k-1}$ in \eqref{eqn:kf-cost} by the empirical prediction mean $\hat{\mathbf{m}}_{k}$ in \eqref{eqn:enkf-mean} and the empirical prediction covariance $\hat{\mathbf{C}}_{k}$ in \eqref{eqn:enkf-cov}. Therefore, our modified cost function for ensembles $\mathbf{x}_{k\mid k}^{(i)}$ for $i=1, 2, \dots, N$ is given by
\begin{equation}\label{eqn:mc-enkf-cost}
\begin{aligned}
    \mathcal{L}_{2}(\mathbf{x}_{k}) & = \mathcal{G}_{\sigma}\left(\|\mathbf{y}_{k} - \mathbf{h}_{k}(\mathbf{x}_{k})\|_{\mathbf{R}_{k}}\right) \\
    & + \mathcal{G}_{\sigma}\left(\|\mathbf{x}_{k} - \hat{\mathbf{m}}_{k}\|_{\hat{\mathbf{C}}_{k}}\right).
\end{aligned}
\end{equation}
Then based on $\mathcal{L}_{2}(\mathbf{x}_{k})$, the update step for MC-EnKF can be obtained  solving the following optimization
problem for $i=1, 2, \dots, N$:
\begin{equation}\label{eqn:opt-form}
    \mathbf{x}_{k\mid k}^{(i)} = \arg\max_{\mathbf{x}_{k}} \ \mathcal{L}_{2}(\mathbf{x}_{k}).
\end{equation}
In what follows, the derivation contains some approximations. For the sake of obtaining our algorithm, we heuristically treat them as exact equalities. Let us denote
\begin{equation}\label{eqn:two-weights}
    \begin{aligned}
     l_{\mathbf{R}_{k}} & =\mathcal{G}_{\sigma}\left(\left\|\mathbf{y}_{k}-\mathbf{h}_{k}\left(\mathbf{x}_{k}\right)\right\|_{\mathbf{R}_{k}}\right) \\
     l_{\hat{\mathbf{C}}_{k}} & =\mathcal{G}_{\sigma}\left(\| \mathbf{x}_{k}-\hat{\mathbf{m}}_{k} \|_{\hat{\mathbf{C}}_{k}}\right).
    \end{aligned}
\end{equation}
Then recall $\mathbf{H}_{k}$ defined in \eqref{eqn:approx-H}, we consider this approximation when taking the derivative,
\begin{equation}
\begin{aligned}
\frac{\partial \mathcal{L}_{2}\left(\mathbf{x}_{k}\right)}{\partial \mathbf{x}_{k}}= & -\frac{1}{\sigma^2}\left(\frac{\partial \mathbf{h}_{k}}{\partial \mathbf{x}_{k}}\right)^{\top} l_{\mathbf{R}_{k}} \mathbf{R}_{k}^{-1}\left(\mathbf{y}_{k}-\mathbf{h}_{k}\left(\mathbf{x}_{k}\right)\right) \\
& +\frac{1}{\sigma^2} l_{\hat{\mathbf{C}}_{k}} \hat{\mathbf{C}}_{k}^{-1}\left(\mathbf{x}_{k}-\hat{\mathbf{m}}_{k}\right) \\
\approx & -\frac{1}{\sigma^2} \mathbf{H}_{k}^{\top} l_{\mathbf{R}_{k}} \mathbf{R}_{k}^{-1}\left(\mathbf{y}_{k}-\mathbf{h}_{k}\left(\mathbf{x}_{k}\right)\right) \\
& + \frac{1}{\sigma^2} l_{\hat{\mathbf{C}}_{k}} \hat{\mathbf{C}}_{k}^{-1}\left(\mathbf{x}_{k}-\hat{\mathbf{m}}_{k}\right).
\end{aligned}
\end{equation}
Now letting $\frac{\partial \mathcal{L}_{2}\left(\mathbf{x}_{k}\right)}{\partial \mathbf{x}_{k}}=0$, we have
\begin{equation}\label{eqn:one-order-condition}
\mathbf{H}_{k}^{\top} l_{\mathbf{R}_{k}} \mathbf{R}_{k}^{-1}\left(\mathbf{y}_{k}-\mathbf{h}_{k}\left(\mathbf{x}_{k}\right)\right)=l_{\hat{\mathbf{C}}_{k}} \hat{\mathbf{C}}_{k}^{-1}\left(\mathbf{x}_{k}-\hat{\mathbf{m}}_{k}\right) .
\end{equation}
Adopting the first-order Taylor series to approximate the nonlinear observation function $\mathbf{h}_{k}$ at $\mathbf{x}_{k\mid k-1}^{(i)}$, i.e., 
\begin{equation}\label{eqn:tarlor-approx-h}
\mathbf{h}_{k}\left(\mathbf{x}_{k}\right) \approx \mathbf{h}_{k}\left(\mathbf{x}_{k\mid k-1}^{(i)}\right)+\mathbf{H}_{k}\left(\mathbf{x}_{k}-\mathbf{x}_{k\mid k-1}^{(i)}\right) .
\end{equation}
Substituting \eqref{eqn:tarlor-approx-h} into \eqref{eqn:one-order-condition}, we have
\begin{equation}\label{eqn:mc-enkf-update}
\begin{aligned}
\mathbf{x}_{k}= & \mathbf{x}_{k\mid k-1}^{(i)}+\left(l_{\hat{\mathbf{C}}_{k}} \hat{\mathbf{C}}_{k}^{-1}+\mathbf{H}_{k}^{\top} l_{\mathbf{R}_{k}} \mathbf{R}_{k}^{-1} \mathbf{H}_{k}\right)^{-1} \\
& \times \mathbf{H}_{k}^{\top} l_{\mathbf{R}_{k}} \mathbf{R}_{k}^{-1}\left(\mathbf{y}_{k}-\mathbf{h}_{k}\left(\mathbf{x}_{k\mid k-1}^{(i)}\right)\right) .
\end{aligned}
\end{equation}
Then we obtain the following stochastic ensemble update rule like \eqref{eqn:enkf-update} for $\mathbf{x}_{k \mid k}^{(i)}$ with drawing $\mathbf{v}_{k}^{(i)}\sim\mathcal{N}(0, \mathbf{R}_{k})$:
\begin{equation}\label{eqn:gain-update}
\mathbf{x}_{k \mid k}^{(i)}=\mathbf{x}_{k\mid k-1}^{(i)}+\tilde{\mathbf{K}}_{k}\left(\mathbf{y}_{k} + \mathbf{v}_{k}^{(i)} -\mathbf{h}_{k}\left(\mathbf{x}_{k\mid k-1}^{(i)}\right)\right),
\end{equation}
where the new Kalman gain $\tilde{\mathbf{K}}_{k}$ is given by
\begin{equation}\label{eqn:enkf-gain-}
\begin{aligned}
 \tilde{\mathbf{K}}_{k}& =\left(l_{\hat{\mathbf{C}}_{k}}\hat{\mathbf{C}}_{k}^{-1}+\mathbf{H}_{k}^{\top} l_{\mathbf{R}_{k}} \mathbf{R}_{k}^{-1} \mathbf{H}_{k}\right)^{-1} \mathbf{H}_{k}^{\top} l_{\mathbf{R}_{k}} \mathbf{R}_{k}^{-1}\\
 & = \left(\left(\frac{\hat{\mathbf{C}}_{k}}{l_{\hat{\mathbf{C}}_{k}}}\right)^{-1}+\mathbf{H}_{k}^{\top}  \left(\frac{\mathbf{R}_{k}}{l_{\mathbf{R}_{k}}}\right)^{-1} \mathbf{H}_{k}\right)^{-1} \mathbf{H}_{k}^{\top} \left(\frac{\mathbf{R}_{k}}{l_{\mathbf{R}_{k}}}\right)^{-1}.
\end{aligned}
\end{equation}
Here we shall discuss the practical algorithm, since $l_{\mathbf{R}_{k}}$ and $l_{\hat{\mathbf{C}}_{k}}$ all contain the item $\mathbf{x}_{k}$.
Here we introduce a novel technical approximation. 
Recall \eqref{eqn:two-weights}, we use $\hat{\mathbf{m}}_{k}$ to approximate $\mathbf{x}_{k}$ contained in $l_{\mathbf{R}_{k}}$ and $l_{\hat{\mathbf{C}}_{k}}$, respectively, i.e., 
\begin{equation}\label{eqn:approx-two-weights}
    \begin{aligned}
        \hat{l}_{\mathbf{R}_{k}} & =  \mathcal{G}_{\sigma}\left(\left\|\mathbf{y}_{k}-\mathbf{h}_{k}\left(\hat{\mathbf{m}}_{k}\right)\right\|_{\mathbf{R}_{k}}\right) \approx l_{\mathbf{R}_{k}}\\
        \hat{l}_{\hat{\mathbf{C}}_{k}} & = \mathcal{G}_{\sigma}\left(\| \hat{\mathbf{m}}_{k} -\hat{\mathbf{m}}_{k} \|_{\hat{\mathbf{C}}_{k}}\right) \approx l_{\hat{\mathbf{C}}_{k}}.
    \end{aligned}
\end{equation}
It follows that we summarise the specific algorithm steps of MC-EnKF in Algorithm \ref{alg:mc-enkf}.
\begin{algorithm}[htbp]
\caption{MC-EnKF}
\label{alg:mc-enkf}
\begin{algorithmic}[1]
    \STATE \textbf{Intitialization.} Start with initial filtering ensembles $\mathbf{x}_{0\mid 0}^{(1)}, \dots, \mathbf{x}_{0\mid 0}^{(N)}$. Then, at each time $k=1, 2, \dots$, given filtering ensembles $\mathbf{x}_{k-1\mid k-1}^{(1)}, \dots, \mathbf{x}_{k-1\mid k-1}^{(N)}$, the MC-EnKF carries out the following two steps for $i=1, 2, \dots, N$:
	\STATE \textbf{Prediction Step: } Draw $\mathbf{w}_{k}^{(i)}\sim\mathcal{N}(0, \mathbf{Q}_{k})$ and calculate $\mathbf{x}_{k\mid k-1}^{(i)}$ via \eqref{eqn:enkf-pred}. \;
	\STATE  \textbf{Update Step: }  Draw $\mathbf{v}_{k}^{(i)}\sim\mathcal{N}(0, \mathbf{R}_{k})$ and calculate $\mathbf{x}_{k\mid k}^{(i)}$ by
    $$
    \mathbf{x}_{k\mid k}^{(i)} = \mathbf{x}_{k\mid k-1}^{(i)} + \tilde{\mathbf{K}}_{k} \left(\mathbf{y}_{k} + \mathbf{v}_{k}^{(i)} - \mathbf{h}_{k}\left(\mathbf{x}_{k\mid k-1}^{(i)}\right)\right),
    $$
    where $\tilde{\mathbf{K}}_{k}$ is defined in \eqref{eqn:enkf-gain-} with the approximation \eqref{eqn:approx-two-weights}. \;
\end{algorithmic}
\end{algorithm}

\begin{remark}
The proposed MC-EnKF is still of Kalman type, i.e., it  has a similar recursive
structure as the common EnKF. Moreover, its computational complexity is only slightly higher than the common EnKF, which is verified in later experiments.
\end{remark}

\begin{remark}
In \eqref{eqn:enkf-gain-}, we note that MC-EnKF is potential to handle non-Gaussian observation noises thanks to the two weights, $\hat{l}_{\hat{\mathbf{C}}_{k}}$ and $\hat{l}_{\mathbf{R}_{k}}$, which could adjust $\hat{\mathbf{C}}_{k}$ and $\mathbf{R}_{k}$ adaptively. We shall also verify this argument in later experiments.
\end{remark}

\subsection{Adaptive Strategy for Kernel Bandwidth Selection and Convergence Regarding to Kernel Bandwidth}
In general, the kernel bandwidth (scale parameter) in the cost function balances the convergence rates of the regression model and its robustness \cite{feng2015learning}. Since our algorithm does not involve solving such regression problems, we focus on the how to tune this scale parameter for robustness (with respect to outliers). Intuitively, if the arrived observation $\mathbf{y}_{k}$ contains large outliers, $\|\mathbf{y}_{k} - \mathbf{h}_{k}(\hat{\mathbf{x}}_{k})\|_{2}$ will be large, where $\|\cdot\|_{2}$ denote the $l^{2}$-norm. In this case, we need a smaller $\sigma$ to make our algorithm more robust. Motivated by this intuition, we propose to set $\sigma_{k} = \frac{1}{\|\mathbf{y}_{k} - \mathbf{h}_{k}(\hat{\mathbf{x}}_{k})\|_{2}}$ adaptively, where $\sigma_{k}$ denote the kernel bandwidth used in the $k$-th step for MC-EnKF. As shown in our simulation experiments, this adaptive strategy can achieve a good estimation performance. Additionally, the MC-EnKF will act more and more like the common EnKF algorithm as $\sigma$ increases. In particular, the following convergence theorem holds .

\begin{theorem}\label{th:main}
If the kernel bandwidth $\sigma\to\infty$, the proposed MC-EnKF will converge to the common EnKF.
\end{theorem}

\begin{proof}
See Appendix \ref{appendix:proof}.
\end{proof}

\begin{remark}
This justification in Theorem \ref{th:main} gives the flexibility of MC-EnKF to handle Gaussian noise cases since the common EnKF performs well in this case.
\end{remark}

\section{Experiments}\label{sec:experiments}

In this section, we will conduct performance comparisons of our proposed MC-EnKF with the common EnKF and the maximum liklihood EnKF (ML-EnKF) \cite{zupanski2005maximum}, which optimizes a nonlinear cost function through maximum likelihood. These evaluations will be carried out on various filtering benchmarks that incorporate non-Gaussian noises, i.e., the observation noises with large outliers. In these experiments, we consider $M=100$ independent Monte Carlo runs. In each run, $N=1000$ samples are used to evaluate the MSE of the state. The EnKF and MC-EnKF are implemented by Numpy \cite{harris2020array} and run on Intel(R) Core(TM) i7-10700 CPU @ 2.90GHz.

\subsection{Linear System}
The first benchmark we used here is a linear system, which is given by
\begin{equation}
\begin{aligned}
    \mathbf{x}_{k} & = \left[\begin{array}{cc}
        \cos(\alpha) & \sin(\alpha)  \\
        -\sin(\alpha) & \cos(\alpha)
    \end{array}\right] \mathbf{x}_{k-1}  + \mathbf{w}_{k}, \\
    \mathbf{y}_{k} & =  \left[\begin{array}{cc}
        1 & 1 
    \end{array}\right]\mathbf{x}_{k} + \mathbf{v}_{k},
\end{aligned}
\end{equation}
where $\alpha=\frac{\pi}{18}$, $\mathbf{x}_{0}\sim\mathcal{N}(0, \mathbf{I}_{2})$ and $\mathbf{w}_{k}\sim \mathcal{N}(0, 0.01\mathbf{I}_{2})$. Let $\mathbf{Q}_{1}=0.01 \mathbf{I}_{1}$ be the nominal observation covariance, the observation noises are sampled from the mixture of Gaussian, i.e.,
\begin{equation}
    \mathbf{v}_{k} \sim 0.9 \ \mathcal{N}(0, \mathbf{Q}_{1}) + 0.1 \ \mathcal{N}(0, 100\mathbf{Q}_{1}).
\end{equation}
Table \ref{table:linear} lists the MSEs and average CPU times of EnKF and MC-EnKF
with different kernel bandwidths in this example, where the number of ensembles of EnKF and MC-EnKF are set as $100$. Note that the MC-EnKF outperforms EnKF in most cases and their average CPU times are virtually identical to the common EnKF. This demonstrates that our proposed MC-EnKF has better performance than those of EnKF with nearly no extra cost on CPU time. It means that MC-EnKF is more robust and efficient when we deal with observation noises with large outliers.
In fact, our proposed adaptive strategy (MC-EnKF-Ada) outperforms most different $\sigma$ cases and is only slightly worse than the best case. This suggests that our adaptive strategy works well for linear system.
When we conduct the simulation for large enough $\sigma$, for example, $\sigma=10000$, the result verifies Theorem \ref{th:main} for linear system since we note that the performance of MC-EnKF is almost identical to the EnKF in this case.
\begin{table}[htbp]
\caption{The MSEs comparisons between EnKF and MC-EnkF with different $\sigma$ for linear system.}
\begin{center}
\begin{tabular}{ccc}
\toprule
Methods & MSE & CPU Time \\
\midrule
EnKF & 3.1004 & 0.4676 \\
MC-EnKF-Ada & 2.2344 & 0.4697 \\
MC-EnKF ($\sigma=0.1$) & 2.9655 & 0.4634 \\
MC-EnKF ($\sigma=0.5$) & 2.9247 & 0.4636 \\
MC-EnKF ($\sigma=2$) & 2.1031 & 0.4689\\
MC-EnKF ($\sigma=5$) & \textbf{1.9411} & 0.4629 \\
MC-EnKF ($\sigma=10$) & 2.3073 & 0.4642\\
MC-EnKF ($\sigma=10000$) & 3.1605 & 0.4725\\
\bottomrule
\end{tabular}
\end{center}
\label{table:linear}
\end{table}

In order to better present the comparison results, we choose the case where $\sigma=5$ to illustrate the true state and EnKF estimation and MC-EnKF estimation over time, which is shown in Fig \ref{fig:linear}. Here, the selected dimension is indicated by the - suffix of the legend in the figures; for example, EnKF-1 denotes the EnKF estimation on the first state dimension. This setting will also be used for the following figure. We can clearly see that this MC-EnKF gives accurate estimates but EnKF performs poorly when the observations contain large outliers. 

\begin{figure}[htbp]
    \centering
    \includegraphics[width=1\linewidth]{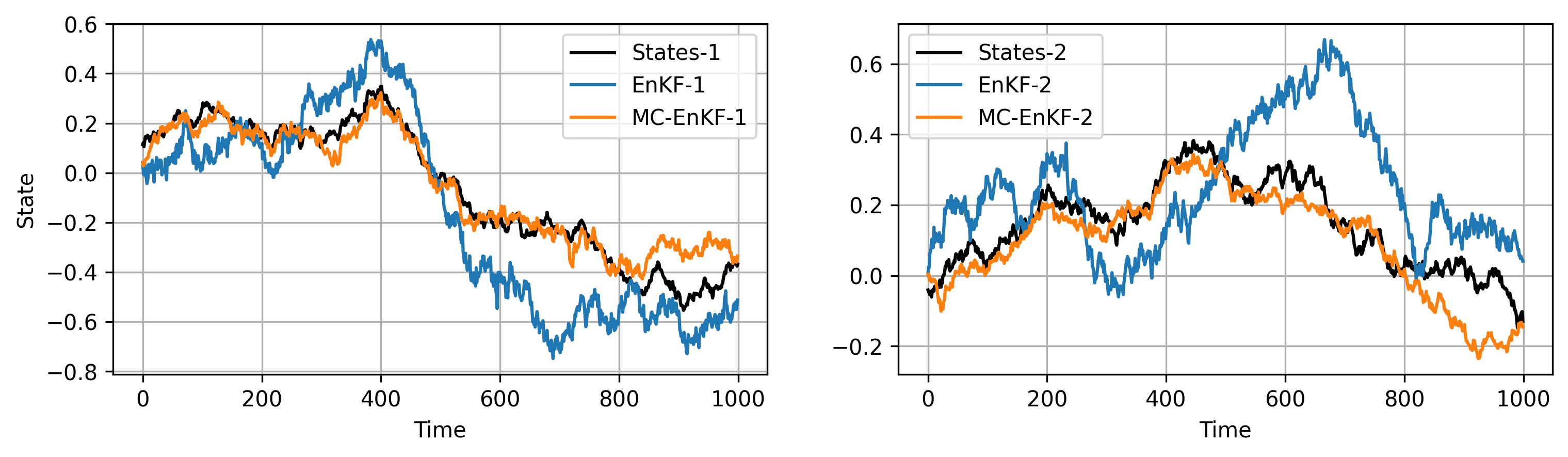}
    \caption{The true state versus the estimate of EnKF and the estimate of MC-EnKF ($\sigma=5$) over time for linear system.}
    \label{fig:linear}
\end{figure}

\subsection{Nonlinear System}
The second benchmark is a nonlinear system, which is given as follows:
\begin{equation}\label{eqn:nonlinear}
    \begin{aligned}
    \mathbf{x}_{k} & =  \left(\mathbf{I}_{2} + \kappa_{1} \left[
\begin{array}{cc}
    -1  &  0.2 \\
    0.2 & -1 
\end{array}
\right] \right) \mathbf{x}_{k-1} + \kappa_{2} \cos(\mathbf{x}_{k-1})  + \mathbf{w}_{k} ,\\
    \mathbf{y}_{k} & = \mathbf{x}_{k} + \sin(\mathbf{x}_{k}) + \mathbf{v}_{k},
    \end{aligned}
\end{equation}
where $\mathbf{x}_{0}\sim\mathcal{N}(0, \mathbf{I}_{2})$ and $\mathbf{w}_{k}\sim \mathcal{N}(0, \mathbf{I}_{2})$. $\kappa_{1} =\kappa_{2}=0.1$ are constants controlling the state dynamics. Let $\mathbf{Q}_{2}=\mathbf{I}_{2}$ be the nominal observation covariance, the observation noises  are sampled from the mixture of Gaussian, i.e.,
\begin{equation}
    \mathbf{v}_{k} \sim 0.9 \ \mathcal{N}(0, \mathbf{Q}_{2}) + 0.1 \ \mathcal{N}(0, 1000\mathbf{Q}_{2}).
\end{equation}
We list the MSEs and average CPU times of EnKF, ML-EnKF and MC-EnKF with different kernel bandwidths in Table \ref{table:nonlinear} for this nonlinear example, where the number of ensembles of these filters still be set as $100$.  It is obvious that for nonlinear system the EnKF degrades greatly in the presence of non-Gaussian observation noises and the linear approximation error of observation function, hence has the worst estimate performance. ML-EnKF performs better than EnKF, but is still affected by outliers. Better results can be obtained with our proposed adaptive strategy (MC-EnKF-Ada) than with some kernel bandwidths, indicating that it is still useful for nonlinear systems. Note that the MC-EnKF outperforms EnKF in all cases, which demonstrate that our proposed MC-EnKF can effectively eliminate the effect of the non-Gaussian observation noises. Additionally, we notice that they share almost identical CPU times, which demonstrates the efficiency of MC-EnKF. Moreover, with suitable kernel bandwidth ($\sigma=5$), MC-EnKF can achieve very accurate estimates while EnKF is heavily affected by large outliers, which is presented
in Fig \ref{fig:nonlinear}. We also note that when $\sigma$ is large enough ($\sigma=10000$), the performance of MC-EnKF is identical to the EnKF, which verifies the result of Theorem \ref{th:main} for nonlinear system.

\begin{table}[htbp]
\caption{The MSEs comparisons between EnKF and MC-EnkF with different $\sigma$ for nonlinear system.}
\begin{center}
\begin{tabular}{ccc}
\toprule
Methods & MSE & CPU Time\\
\midrule
EnKF & 4.0929 & 0.4768 \\
ML-EnKF & 3.3567 & 0.5059 \\
MC-EnKF-Ada & 2.9282 & 0.4775 \\
MC-EnKF ($\sigma=0.1$) & 3.0150 & 0.4793 \\
MC-EnKF ($\sigma=0.5$) & 2.9703 & 0.4769 \\
MC-EnKF ($\sigma=2$) & 1.5849 & 0.4713\\
MC-EnKF ($\sigma=5$) & \textbf{1.3012} & 0.4762 \\
MC-EnKF ($\sigma=10$) & 1.3752 & 0.4729 \\
MC-EnKF ($\sigma=10000$) & 4.0448 & 0.4799\\
\bottomrule
\end{tabular}
\end{center}
\label{table:nonlinear}
\end{table}

\begin{figure}[htbp]
    \centering
    \includegraphics[width=1\linewidth]{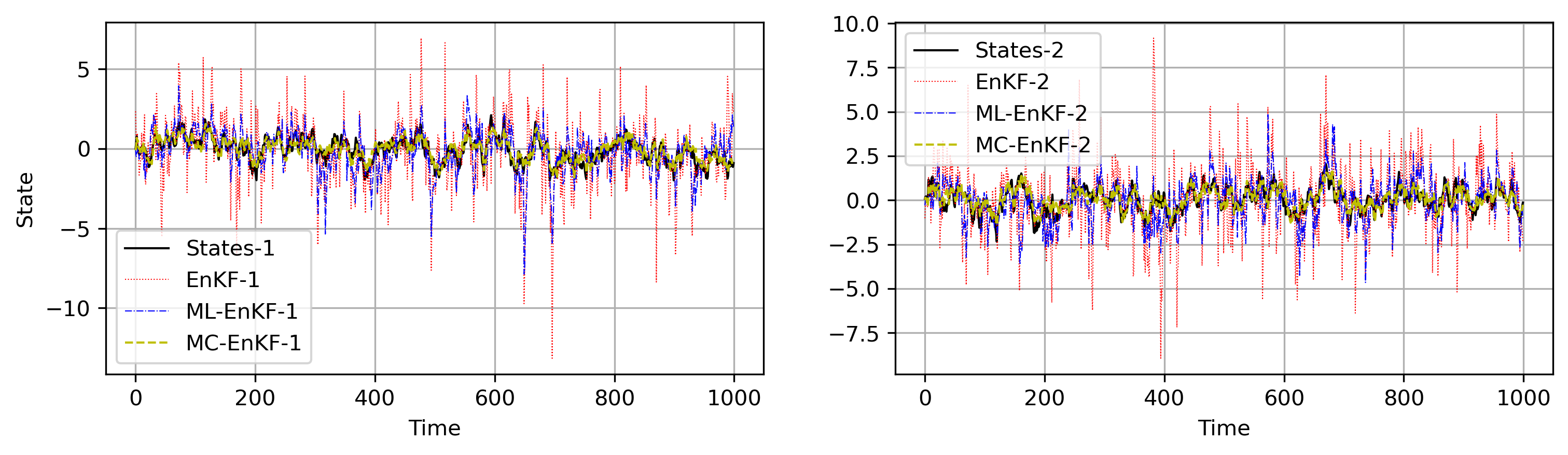}
    \caption{The true state versus the estimate of EnKF and the estimate of MC-EnKF ($\sigma=5$) over time for nonlinear system.}
    \label{fig:nonlinear}
\end{figure}

\subsection{Discussions}
According to the simulation experiments mentioned above, we notice the following two arguments:
\begin{itemize}
    \item With suitable kernel bandwidth, MC-EnKF can significantly improve its robustness to observation noises containing large outliers with nearly no CPU time loss.
    \item With large kernel bandwidth, MC-EnKF performs like the common EnKF, which implies its flexibility to handle Gaussian noise cases.
\end{itemize}
These two arguments make our viewpoint a very fruitful area for further study. Besides,
we note that the kernel bandwidth $\sigma$ plays a significant role in the algorithm and will have an impact on its robustness to non-Gaussian observation (such as outliers). Although our proposed adaptive strategy is effective, it cannot achieve the performance with suitable kernel bandwidth. Therefore, a highly intriguing future work direction may be a more effective adaptive strategy to pick the right kernel bandwidth with theoretical guarantee since it is crucial for our proposed MC-EnKF. In Addition,
this article only investigates using the special case (MCC) of generalized MCC \cite{chen2016generalized} as a cost function to derive a robust EnKF with stochastic update steps. Hence how to develop robust variants of EnKF \cite{evensen2003ensemble} based on the generalized MCC will be our next research focus.

\section{Conclusion}\label{sec:conclusion}

This article proposes a robust EnKF filtering algorithm called maximal correntropy ensemble Kalman filter (MC-EnKF), which is flexible to handle cases involving both Gaussian and non-Gaussian noise. Instead of the well-known minimum mean square error (MMSE) criterion, the MC-EnKF is derived by employing the maximum correntropy criterion (MCC) as the optimality criterion. Ensemble propagation equations remain to be of the Kalman type. The MC-EnKF will behave like the EnKF when the kernel bandwidth is large enough, and we also theoretically demonstrate this argument. With the proper kernel bandwidth, MC-EnKF can perform much better than the EnKF at only a slight increase in computational cost, especially when the underlying observation system is disturbed by some heavy-tailed non-Gaussian noises. Besides, we propose an adaptive strategy to help us choose kernel bandwidth, and its effectiveness is also been verified by experiments.


%

\appendices
\section{Proof of the Theorem \ref{th:main}}\label{appendix:proof}

Here we shall give the proof of the Theorem \ref{th:main}. Before proceeding, we need the following technical lemma.
\begin{lemma}[Matrix Inversion
Lemma \cite{hager1989updating}]\label{lemma:mat-inv-lemma}
If $\mathbf{A}\in\mathbb{R}^{n\times n}$, $\mathbf{C}\in\mathbb{R}^{n\times n}$ are non-singular, $\mathbf{B}\in\mathbb{R}^{n\times m}$, $\mathbf{D}\in\mathbb{R}^{m\times n}$,
\begin{equation}
    \left(\mathbf{A} + \mathbf{B}\mathbf{C}\mathbf{D}\right)^{-1} = \mathbf{A}^{-1} - \mathbf{A}^{-1}\mathbf{B}\left(\mathbf{D}\mathbf{A}^{-1}\mathbf{B} + \mathbf{C}^{-1}\right)^{-1}\mathbf{D}\mathbf{A}^{-1}
\end{equation}
\end{lemma}
\begin{proof}
Now we state the proof of Theorem \ref{th:main}. Note that when the kernel bandwidth $\sigma\to\infty$, 
\begin{equation}
\begin{aligned}
   \hat{l}_{\mathbf{R}_{k}} & =\mathcal{G}_{\sigma}\left(\left\|\mathbf{y}_{k}-\mathbf{h}_{k}\left(\hat{\mathbf{m}}_{k}\right)\right\|_{\mathbf{R}_{k}}\right) \\
   & = \exp\left( \frac{-\left(\mathbf{y}_{k}-\mathbf{h}_{k}\left(\hat{\mathbf{m}}_{k}\right)\right)^{\top}\mathbf{R}_{k}^{-1}\left(\mathbf{y}_{k}-\mathbf{h}_{k}\left(\hat{\mathbf{m}}_{k}\right)\right)}{2\sigma^{2} }\right) \\
   & \to 1,
\end{aligned}
\end{equation}
Similarly, one can conclude that $\hat{l}_{\hat{\mathbf{C}}_{k}}\to 1$ when $\sigma\to\infty$. These arguments imply $\tilde{\mathbf{K}}_{k}$ in \eqref{eqn:enkf-gain-} with the approximation \eqref{eqn:approx-two-weights} become
\begin{equation}
\tilde{\mathbf{K}}_{k}=\left(\hat{\mathbf{C}}_{k}^{-1}+\mathbf{H}_{k}^{\top}  \mathbf{R}_{k}^{-1} \mathbf{H}_{k}\right)^{-1} \mathbf{H}_{k}^{\top}\mathbf{R}_{k}^{-1},
\end{equation}
Then in view of Lemma \ref{lemma:mat-inv-lemma} (we choose $\mathbf{A}=\hat{\mathbf{C}}_{k}^{-1}$, $\mathbf{B}=\mathbf{H}_{k}^{\top} $, $\mathbf{C}=\mathbf{R}_{k}$ and $\mathbf{D}=\mathbf{H}_{k}$), one can conclude that
\begin{equation}
\begin{aligned}
    \tilde{\mathbf{K}}_{k} & =\left(\hat{\mathbf{C}}_{k}^{-1}+\mathbf{H}_{k}^{\top}  \mathbf{R}_{k}^{-1} \mathbf{H}_{k}\right)^{-1} \mathbf{H}_{k}^{\top}\mathbf{R}_{k}^{-1} \\
    & = \left(\hat{\mathbf{C}}_{k} - \hat{\mathbf{C}}_{k}\mathbf{H}_{k}^{\top}\left(\mathbf{R}_{k} + \mathbf{H}_{k}\hat{\mathbf{C}}_{k}\mathbf{H}_{k}^{\top}\right)^{-1}\mathbf{H}_{k}\hat{\mathbf{C}}_{k} \right)\mathbf{H}_{k}^{\top}\mathbf{R}_{k}^{-1} \\
    & = \hat{\mathbf{C}}_{k}\mathbf{H}_{k}^{\top}\mathbf{R}_{k}^{-1} - \hat{\mathbf{C}}_{k}\mathbf{H}_{k}^{\top}\left(\mathbf{R}_{k} + \mathbf{H}_{k}\hat{\mathbf{C}}_{k}\mathbf{H}_{k}^{\top}\right)^{-1} \\
    & \times \mathbf{H}_{k}\hat{\mathbf{C}}_{k}\mathbf{H}_{k}^{\top}\mathbf{R}_{k}^{-1} \\
    & = \hat{\mathbf{C}}_{k}\mathbf{H}_{k}^{\top}\left(\mathbf{R}_{k}^{-1}- \left(\mathbf{R}_{k} + \mathbf{H}_{k}\hat{\mathbf{C}}_{k}\mathbf{H}_{k}^{\top}\right)^{-1}\mathbf{H}_{k}\hat{\mathbf{C}}_{k}\mathbf{H}_{k}^{\top}\mathbf{R}_{k}^{-1}\right) \\
    &  = \hat{\mathbf{C}}_{k}\mathbf{H}_{k}^{\top}\left(\mathbf{I}- \left(\mathbf{R}_{k} + \mathbf{H}_{k}\hat{\mathbf{C}}_{k}\mathbf{H}_{k}^{\top}\right)^{-1}\mathbf{H}_{k}\hat{\mathbf{C}}_{k}\mathbf{H}_{k}^{\top}\right)\mathbf{R}_{k}^{-1}.
\end{aligned}
\end{equation}
Now use the Lemma \ref{lemma:mat-inv-lemma} again (we choose $\mathbf{A}=\mathbf{I}$, $\mathbf{B}=\mathbf{I} $, $\mathbf{C}=\mathbf{R}_{k}$ and $\mathbf{D}=\mathbf{H}_{k}\hat{\mathbf{C}}_{k}\mathbf{H}_{k}^{\top}$), one can further conclude that
\begin{equation}
   \begin{aligned}
    \tilde{\mathbf{K}}_{k} & =\hat{\mathbf{C}}_{k}\mathbf{H}_{k}^{\top}\left(\mathbf{I}- \left(\mathbf{R}_{k} + \mathbf{H}_{k}\hat{\mathbf{C}}_{k}\mathbf{H}_{k}^{\top}\right)^{-1}\mathbf{H}_{k}\hat{\mathbf{C}}_{k}\mathbf{H}_{k}^{\top}\right)\mathbf{R}_{k}^{-1} \\
    & = \hat{\mathbf{C}}_{k}\mathbf{H}_{k}^{\top}\left(\mathbf{I} + \mathbf{R}_{k}^{-1}\mathbf{H}_{k}\hat{\mathbf{C}}_{k}\mathbf{H}_{k}^{\top} \right)^{-1}\mathbf{R}_{k}^{-1} \\
    & = \hat{\mathbf{C}}_{k}\mathbf{H}_{k}^{\top}\left(\mathbf{R}_{k} +\mathbf{H}_{k}\hat{\mathbf{C}}_{k}\mathbf{H}_{k}^{\top} \right)^{-1},
    \end{aligned}
\end{equation}
which is equal to \eqref{eqn:enkf-gain}.
\end{proof}


\bibliographystyle{IEEEtran}
\bibliography{ref}

\end{document}